\documentclass{article}
\usepackage{fullpage}
\usepackage{amsthm}
\usepackage{hyperref}
\newtheorem{theorem}{Theorem}

\newtheorem{cor}{Corollary}
\newtheorem{definition}{Definition}

\usepackage{graphicx,amssymb,amsmath}

\usepackage{colortbl}
\usepackage{xspace}
\usepackage{subfigure}
\usepackage{psfrag}       %latex text in .eps files

\newcommand{\RNs}{{Radio Networks}\xspace}
\newcommand{\RG}{RGG\xspace}
\newcommand{\mig}[1]{{#1}}

\title{Probabilistic Bounds on the Length of a Longest Edge in Delaunay Graphs of Random Points in $d$-Dimensions
\thanks{This research was partially supported by 
Spanish MICINN grant TIN2008-06735-C02-01, 
Comunidad de Madrid grant S2009TIC-1692, 
EU Marie Curie International Reintegration Grant IRG 210021, 
and the National Science Foundation (CCF-0937829, CCF-1018388). 
This work has been presented in~\cite{AFMM:delaunayCCCG},
and an earlier version has been presented in~\cite{AFMM:delaunayFWCG}.
}
}

\author{
Esther M. Arkin~\thanks{Department of Applied Math and Statistics, Stony Brook University, USA. \href{mailto:esther.arkin@stonybrook.edu}{esther.arkin@stonybrook.edu}, \href{mailto:joseph.mitchell@stonybrook.edu}{joseph.mitchell@stonybrook.edu}}
\\\\
Joseph S. B. Mitchell~\footnotemark[1]
\and
Antonio Fern\'andez Anta~\thanks{Institute IMDEA Networks, Spain. \href{mailto: antonio.fernandez@imdea.org}{ antonio.fernandez@imdea.org}}
\\\\
Miguel A. Mosteiro~\thanks{Computer Science Department, Rutgers University, USA. \href{mailto:mosteiro@cs.rutgers.edu}{mosteiro@cs.rutgers.edu}}~\thanks{LADyR, GSyC, Universidad Rey Juan Carlos, Spain.}
}

\begin{document}
\thispagestyle{empty}
\maketitle

\begin{abstract}
Motivated by low energy consumption in geographic routing in wireless
networks, there has been recent interest in determining bounds on the
length of edges in the Delaunay graph of randomly distributed points.
Asymptotic results are known for random networks in planar domains.
In this paper, we obtain upper and lower bounds that hold with
parametric probability in any dimension, for points distributed
uniformly at random in domains with and without boundary.  The results
obtained are asymptotically tight for all relevant values of such
probability and constant number of dimensions, and show that the
overhead produced by boundary nodes in the plane holds also for higher
dimensions.  To our knowledge, this is the first comprehensive study
on the lengths of long edges in Delaunay graphs.
\end{abstract}

% input{intro}
\section{Introduction}

We study the length of a longest Delaunay edge for points randomly
distributed in multidimensional Euclidean spaces.  In particular, we
consider the Delaunay graph for a set of $n$ points distributed
uniformly at random in a $d$-dimensional body of unit volume.  It is
known that the probability that uniformly distributed random points
are not in general position~\footnote{
A set of $d+1$ points in $d$-dimensional Euclidean space is said to be
\emph{in general position} if no hyperplane contains all of them. We
say that such a set is \emph{generic}, or \emph{degenerate}
otherwise.} 
% xxx do we also need to assume that no $d+2$ points are cospherical??
%
is negligible and therefore it is safe to focus on generic
sets of points~\cite{compgeom:2000}, which we do throughout the paper.

The motivation to study such settings comes from the Random Geometric
Graph (\RG) model in which
%
% $\mathcal{G}_{n,r}$, 
%
$n$ nodes are distributed uniformly at random in a disk or, more
generally, according to some specified density function on
$d$-dimensional Euclidean space~\cite{rgg}.  The problem has attracted
recent interest because of its applications in energy-efficient
geometric routing and flooding in wireless sensor networks (see, e.g.,
\cite{bose1999online,kozma:longestDelaunay,kranakis1999compass,lebhar2009unit}).

\paragraph{Related Work.}
Kozma, Lotker, Sharir, and Stupp~\cite{kozma:longestDelaunay} show
that the asymptotic length of a longest Delaunay edge depends on the
sum, $\sigma$, of the distances to the boundary of its endpoints.
%, varying between
%$O(\sqrt[3]{(\log n)/n})$ close to the boundary (specifically, if the
%sum of the distances to the boundary is at most $((\log n)/n)^{2/3}$)
%and $O(\sqrt{(\log n)/n})$ further interior (specifically, if the sum
%of the distances to the boundary is at most $((\log n)/n)^{2/3}$).
%
More specifically, their bounds are 
$O(\sqrt[3]{(\log n)/n})$ if $\sigma\leq ((\log n)/n)^{2/3}$,
$O(\sqrt{(\log n)/n})$ if $\sigma\geq\sqrt{(\log n)/n}$,
and $O((\log n)/(n\sigma))$ otherwise.
Kozma et al. also show, in the same setting, that the expected sum of
the squares of all Delaunay edge lengths is $O(1)$.  
% added for this version
In~\cite{BEY:expDelaunay} the authors consider the Delaunay triangulation of an infinite random (Poisson) point set in $d$ dimensional space. In particular, they study different properties of the subset of those Delaunay edges completely included in a cube $[0, n^{1/d}]\times\dots\times[0, n^{1/d}]$. For the maximum length of a Delaunay edge in this setting, they observe that in expectation is in $\Theta(\log^{1/d} n)$.

The lengths of longest edges in geometric graphs induced by random point sets has
also been studied for graphs related to the Delaunay, including
Gabriel graphs~\cite{wan2007longest} and relative neighborhood (RNG)
graphs~\cite{wan2008longest,yi2010sharp}.  In particular, Wan and
Yi~\cite{wan2007longest} show that for $n$ points uniformly
distributed in a unit-area disk, the ratio of the length of a longest
Gabriel edge to $\sqrt{(\ln n)/(\pi n)}$ is asymptotically almost
surely equal to 2, and the expected number of ``long'' Gabriel edges,
of length at least $2\sqrt{(\ln n+\xi)/(\pi n)}$, is asymptotically
almost surely equal to $2e^{-\xi}$, for any fixed $\xi$.
%\textcolor{red}{
In~\cite{DGM:gabrielarXiv}, while studying the maximum degree of Gabriel and Yao graphs, the authors observe that the probability that the maximum edge length is greater than 
$3\sqrt{(\log n) /n}$ tends to zero, 
%and it is $O(((\log n) /n)^{1/d})$, with high probability, for dimension $d$.
bound that they claim becomes $O(((\log n) /n)^{1/d})$ for $d$ dimensions.
An overview of related problems can be found in~\cite{AS:spatialNets}.
%}

Interest in bounding the length of a longest Delaunay edge in
two-dimensional spaces has grown out of extensive algorithmic
work~\cite{bose:localDelaunay,avin:restrDelaunay,araujo:singleStepDelaunay}
aimed at reducing the energy consumption of geographically routing
messages in \RNs.  Multidimensional Delaunay graphs are well studied
in computational geometry from the point of view of efficient
algorithms to construct them (see \cite{compgeom:2000} and references
therein), but only limited results are known regarding probabilistic
analysis of Delaunay graphs in higher
dimensions~\cite{lemaire:constDelaunay}.

%% xxx check refs; say more??

\paragraph{Overview of Our Results.}
We study the probabilistic length of longest Delaunay edges for points
distributed in geometric domains in two and more dimensions.  Since
the length of the longest Delaunay edge is strongly influenced by the
boundary of the enclosing region, we study the problem for two cases,
which we call \emph{with boundary} and \emph{without boundary}.

Our results include upper and lower bounds for $d$-dimensional bodies
with and without boundaries, that hold for a parametric error
probability $\varepsilon$ 
%%%%%%% Miguel edit:
and are computed up to the constant factors (they are tight only asymptotically).
%These bounds are tight up to constants (they only match asymptotically and for $d\in O(1)$).  
In comparison, the upper bounds presented in~\cite{kozma:longestDelaunay} are only
asymptotic,
%%%%%%% Miguel edit:
% (for large $n$), 
are restricted to two dimensions ($d=2$),
and apply to domains with boundary (disks), although results without
boundary are implicitly given, since the results are parametric in the
distance to the boundary.

Lower bounds without boundary and all upper bounds apply for any
$d>1$.  Lower bounds with boundary are shown for $d\in\{2,3\}$.  The
results shown are asymptotically tight for $e^{-cn}\leq
\varepsilon\leq n^{-c}$, for any constant $c>0$, and $d\in O(1)$.  To
the best of our knowledge, this is the first comprehensive study of
this problem.  The results obtained are summarized in
Table~\ref{table}.  In order to compare upper and lower bounds for
bodies with boundary, it is crucial to notice that we bound the volume
of a circular segment (2D) and the volume of an spherical cap (3D),
which can be approximated by polynomials of third and fourth degree
respectively on the diameter of the base.
%%%%%%%% Antonio edit:
Upper bounds are proved exploiting the fact that, thanks to the
uniform density, it is very unlikely that a ``large'' volume is void
of points.  Lower bounds, on the other hand, are proved by showing that
a configuration that yields a Delaunay edge of a certain length is
not very unlikely.

In the following section, some necessary notation is introduced.
Upper and lower bounds for enclosing bodies without boundaries are
shown in Section~\ref{section:wob}, and the case with boundaries is
covered in Section~\ref{section:wb}.  We conclude with some open
problems.

\begin{table*}
\centering
\begin{tabular}{|p{.7in}|c|c|c|}
\hline
&
~$d$~~&
\begin{tabular}{l}Upper Bound:\\w.p. $\geq 1-\varepsilon$, $\nexists~\widehat{ab}\in D(P)$\end{tabular}&
\begin{tabular}{l}Lower Bound:\\w.p. $\geq \varepsilon$, $\exists~\widehat{ab}\in D(P)$\end{tabular}\\[.1in]
\hline
\rule{0pt}{6ex}\begin{tabular}{l}Without\\boundary\end{tabular}&
\cellcolor[gray]{.8}$d$&
\cellcolor[gray]{.8}$A_{d}(\delta(a,b))\geq \frac{\ln\left(\binom{n}{2}\mig{\binom{n-2}{d-1}}\big/\varepsilon\right)}{\mig{n-d-1}}$&
\cellcolor[gray]{.8}$A_{d}(\delta(a,b)) \geq \frac{\ln\left((e-1)/(e^2\varepsilon)\right)}{n-2+\ln\left((e-1)/(e^2\varepsilon)\right)}$\\[.1in]
\rule{0pt}{6ex}&
\cellcolor[gray]{.9}$1$&
\cellcolor[gray]{.9}$\delta(a,b) \geq \frac{\ln\left(\binom{n}{2}\big/\varepsilon\right)}{n-2}$&
\cellcolor[gray]{.9}$\delta(a,b) \geq \frac{\ln\left((e-1)/(e^2\varepsilon)\right)}{n-2+\ln\left((e-1)/(e^2\varepsilon)\right)}$\\[.1in]
\rule{0pt}{8ex}&
\cellcolor[gray]{.8}$2$&
\cellcolor[gray]{.8}
$\delta(a,b) \geq \frac{\cos^{-1} \left(1-\frac{2\ln\left(\binom{n}{2}\mig{(n-2)}/\varepsilon\right)}{\mig{n-3}}\right)}{\sqrt{\pi}}$&
\cellcolor[gray]{.8}$\delta(a,b) \geq \frac{\cos^{-1}\left(1-\frac{2\ln\left((e-1)/(e^2\varepsilon)\right)}{n-2+\ln\left((e-1)/(e^2\varepsilon)\right)}\right)}{\sqrt{\pi}}$\\[.1in]
\hline
\rule{0pt}{6ex}\begin{tabular}{l}With\\boundary\end{tabular}&
\cellcolor[gray]{.9}$d$&
\cellcolor[gray]{.9}$V_{d}(d(a,b))\geq \frac{\ln\left(\binom{n}{2}\mig{\binom{n-2}{d-1}}\big/\varepsilon\right)}{\mig{n-d-1}}$&
\cellcolor[gray]{.9}--\\[.1in]
\rule{0pt}{7ex}&
\cellcolor[gray]{.8}$2$&
\cellcolor[gray]{.8}$d(a,b) \geq \sqrt[3]{\frac{16}{\sqrt{\pi}}\frac{\ln\left(\binom{n}{2}\mig{(n-2)}\big/\varepsilon\right)}{\mig{n-3}}}$&
\cellcolor[gray]{.8}
\begin{tabular}{c}
$d(a,b) \geq \rho_2/2$~:~$V_{2}(\rho_2) = \frac{\ln\left(\alpha_2/\varepsilon\right)}{\left(n-2+\ln\left(\alpha_2/\varepsilon\right)\right)}$\\[.1in]
\mig{$\implies d(a,b) \geq \sqrt[3]{ \frac{\ln\left(\alpha/\varepsilon\right)}{2\sqrt{\pi}\left(n-2+\ln\left(\alpha/\varepsilon\right)\right)}}$}
\end{tabular}\\[.2in]
\rule{0pt}{7ex}&
\cellcolor[gray]{.9}$3$&
\cellcolor[gray]{.9}$d(a,b) \geq \sqrt[4]{\frac{96}{\pi^{3/2}}\frac{\ln\left(\binom{n}{2}\mig{\binom{n-2}{2}}\big/\varepsilon\right)}{\mig{n-4}}}$&
\cellcolor[gray]{.9}
\begin{tabular}{l}
$d(a,b) \geq \rho_3/2$~:~$V_{3}(\rho_3) = \frac{\ln\left(\alpha_3/\varepsilon\right)}{\left(n-2+\ln\left(\alpha_3/\varepsilon\right)\right)}$\\[.1in]
\mig{$\implies d(a,b) \geq \sqrt[4]{\sqrt[3]{\frac{48}{\pi^4}} \frac{\ln\left(\alpha_3/\varepsilon\right)}{\left(n-2+\ln\left(\alpha_3/\varepsilon\right)\right)}}$}
\end{tabular}\\[.2in]
\hline
\end{tabular}
\caption{Summary of results. $\alpha_2,\alpha_3$ are constants.}  %% say more in caption?? xxx
\label{table}
\end{table*}

%%%%%%%%%%%%%%%%%%%%%%%%%%%%%%%%%%%%%%%%%

%% \input{prelim}

\section{Preliminaries}
%\paragraph{Preliminaries}
The following notation will be used throughout.  We will restrict
attention to Euclidean ($L_2$) spaces.  A \emph{$d$-sphere},
$S=S_{r,c}$, of radius $r$ is the set of all points in a
$(d+1)$-dimensional space 
that are located at distance $r$ (the
\emph{radius}) from a given point $c$ (the \emph{center}).  A
\emph{$d$-ball}, $B=B_{r,c}$, of radius $r$ is the set of all points
in a $d$-dimensional space that are located at distance \emph{at
most} $r$ (the \emph{radius}) from a given point $c$ (the
\emph{center}).  The \emph{area} of a $d$-sphere $S$ (in $(d+1)$-space) is its
$d$-dimensional volume.  The \emph{volume} of a $d$-ball $B$ (in $d$-space) is its
$d$-dimensional volume.  We refer to a \emph{unit sphere} as a sphere
of area $1$ and a \emph{unit ball} as a ball of volume $1$.  (This is
in contrast with some definitions of a ``unit'' ball/sphere as a
unit-radius ball/sphere; we find it convenient to standardize the
volume/area to be $1$ in all dimensions.)
%
% xxx mention here what the radius is of a unit ball/sphere in d-dimensions??

Let $P$ be a set of points on a $d$-sphere, $S$. 
%
%\footnote{We assume familiarity with spherical geometry.} 
%
Given two points $a,b\in P$, let $\widehat{ab}$ be the arc of a great
circle between them.  Let $\delta(a,b)$ be the length of the arc
$\widehat{ab}$, which is also known as the \emph{orthodromic distance}
between $a$ and $b$ on the sphere $S$.  Let the \emph{orthodromic
diameter} of a subset $X\subseteq S$ be the greatest orthodromic
distance between a pair of points in $X$.  A {\em spherical cap on
$S$} is the set of all points at orthodromic distance at most $r$ from
some center point $c\in S$.  Let $A_{d}(x)$ be the area ($d$-volume)
of a spherical cap of orthodromic diameter $x$, on a $d$-sphere of
surface area $1$.  A {\em ball cap of $B$} is the intersection of a
$d$-ball $B$ with a closed halfspace, bounded by a hyperplane $h$, in
$d$-space; the {\em base} of a ball cap is the $(d-1)$-ball that is
the intersection of $h$ with the ball $B$.  Let $V_{d}(x)$ be the
$d$-volume of a ball cap of base diameter $x$, of a $d$-ball of volume
$1$.  
For any pair of points $a,b$, let $d(a,b)$ be the Euclidean distance between $a$ and $b$, i.e. $d(a,b)=||\overrightarrow{ab}||_2$.
Let $D(P)$ be the Delaunay graph of a set of points $P$.

The following definitions of a Delaunay graph, $D(P)$, of a finite set
$P$ of points in $d$-dimensional bodies follow the standard
definitions of Delaunay graphs (see, e.g., Theorem 9.6
in~\cite{compgeom:2000}).

\begin{definition}
\label{dsphDel}
Let $P$ be a generic set of points on a $d$-sphere~$S$. 
\begin{itemize}
\item [(i)] 
A set $F\subseteq P$ of $d+1$ points define the vertices of a {\em Delaunay face} of $D(P)$ if and only if 
there is a $d$-dimensional spherical cap $C\subset S$ such that 
$F$ is contained in the boundary, $\partial C$, of $C$ 
and no points of $P$ lie in the interior of $C$ (relative to the sphere $S$).
\item [(ii)] Two points $a,b\in P$ form a {\em Delaunay edge}, an arc
of $D(P)$, if and only if there is a $d$-dimensional spherical cap $C$
such that $a,b\in \partial C$ and no points of $P$ lie in the interior
of $C$ (relative to the sphere $S$).
\end{itemize}
\end{definition}

\begin{definition}
\label{dballDel}
Let $P$ be a \mig{generic} set of points in a $d$-ball~$B$.  
\begin{itemize}
\item [(i)] 
A set $F\subseteq P$ of $d+1$ points define the vertices of a {\em Delaunay face} of $D(P)$ if and only if 
there is a $d$-ball $B'$ such that 
%%%%%%%% Miguel edit:
%$F\subset\partial B'$ 
$F$ is contained in the boundary, $\partial B'$, of $B'$
and no points of $P$ lie in the interior of $B'$.
\item [(ii)] 
Two points $a,b\in P$
form a {\em Delaunay edge}, an arc of $D(P)$, if and only if there is a $d$-ball $B'$ such
that $a,b \in \partial B'$ and no points of $P$ lie in the interior of $B'$.
\end{itemize}
\end{definition}

The following inequalities~\cite{book:mitrinovic} are used throughout
\begin{align}
e^{-x/(1-x)} \leq 1-x \leq e^{-x}, \textrm{ for } 0<x<1.\label{taylor}
\end{align}

%%%%%%%%%%%%%%%%%%%%%%%%%%%%%%%%%%%%%%%%%

\section{Enclosing Body without Boundary}
\label{section:wob}

The following theorems show upper and lower bounds on the length of arcs in the Delaunay graph on a $d$-sphere. 

%%%%%%%%%%%%%%%%%%%%%%%%%%%%%%%%%%%%%
\subsection{Upper Bound}

\begin{theorem}
\label{thm:dsphereub}
Consider the Delaunay graph $D(P)$ of a set $P$ of $n>d+1\geq 2$ points distributed 
uniformly and independently at random in a unit $d$-sphere, $S$.  Then, for $0< \varepsilon < 1$, the probability is at least $1-\varepsilon$ that
there is no arc $\widehat{ab}\in D(P)$, $a,b\in P$, such that 
$$A_{d}(\delta(a,b))\geq \frac{\ln\left(\binom{n}{2}\mig{\binom{n-2}{d-1}}\big/\varepsilon\right)}{\mig{n-d-1}}.\qquad \qquad (*)$$
\end{theorem}

\begin{proof}
Let $E_\varepsilon$ be the event that ``there exists an arc
$\widehat{ab}\in D(P)$, $a,b\in P$, with inequality $(*)$ satisfied''
Our goal is to prove that $P(E_\varepsilon) \leq \varepsilon$.

Let us consider a fixed pair of points, $a,b\in P$.  We let $E_{a,b}$
be the event that $\widehat{ab}\in D(P)$.  For any subset $Q\subset P$
of $d+1$ points containing $a$ and $b$, let $C_Q$ denote the spherical
cap through $Q$ and let $F_Q$ denote the event that the interior of
$C_Q$ contains no points of $P$ (i.e., $int(C_Q)\cap P=\emptyset$).

Thus, we can write $E_{a,b}=\bigcup_Q F_Q$ as the union, over all
$\binom{n-2}{(d+1)-2}=\binom{n-2}{d-1}$ subsets $Q\subset P$ with
$|Q|=d+1$ and $a,b\in Q$, of the events $F_Q$.  Then, by the union
bound, we know that $P(E_{a,b})\leq \sum_Q P(F_Q)$.  Further, in order
for event $F_Q$ to occur, all points of $P$ {\em except} the $d+1$
points of $Q$ must lie {\em outside} the spherical cap $C_Q$ through
$Q$; thus, $P(F_Q)=(1-\mu_d(C_Q))^{n-(d+1)}$, where $\mu_d(C_Q)$
denotes the $d$-volume of $C_Q$.

We see that $P(F_Q)\leq (1-A_d(\delta(a,b)))^{n-(d+1)}$, since, for
any subset $Q\supset \{a,b\}$, the $d$-volume $\mu_d(C_Q)$ is at least
as large as the $d$-volume, $A_d(\delta(a,b))$, of the spherical cap
having orthodromic diameter $\delta(a,b)$.  (In other words,
$A_d(\delta(a,b))$ is the $d$-volume of the smallest volume spherical
cap whose boundary passes through $a$ and $b$.)

Altogether, we get 
\begin{align*}
P(E_{a,b}) &\leq \sum_Q P(F_Q) = \sum_Q (1-\mu_d(C_Q))^{n-(d+1)}\\ 
           &\leq \binom{n-2}{d-1}(1-A_d(\delta(a,b)))^{n-(d+1)}.
\end{align*}

Now, the event of interest is
$$E_\varepsilon = \bigcup_{a,b\in P: (*) \textrm{ holds}} E_{a,b}.$$ 
The inequality $(*)$ is equivalent to
$$(n-d-1)A_d(\delta(a,b)) \geq \ln\left(\binom{n}{2}\binom{n-2}{d-1}\big/\varepsilon\right),$$
which is equivalent to
$$\left(e^{-A_d(\delta(a,b))}\right)^{(n-d-1)} \leq \frac{\varepsilon}{\binom{n}{2}\binom{n-2}{d-1}}.$$
Since, by Inequality~\ref{taylor}, $e^{-x}\geq 1-x$, the above inequality implies that
$$\left(1-A_d(\delta(a,b))\right)^{(n-d-1)} \leq \frac{\varepsilon}{\binom{n}{2}\binom{n-2}{d-1}},$$
which implies that
$$\binom{n}{2}\binom{n-2}{d-1}\left(1-A_d(\delta(a,b))\right)^{(n-d-1)} \leq \varepsilon.$$

Using the union bound, we get
$$P(E_\varepsilon) = P\left(\bigcup_{a,b\in P: (*) \textrm{ holds}} E_{a,b}\right) \leq \sum_{a,b\in P: (*) \textrm{ holds}} P(E_{a,b}).$$
Since each term $P(E_{a,b})$ in the above summation is bounded above by $\binom{n-2}{d-1}(1-A_d(\delta(a,b)))^{n-(d+1)}$,
and there are at most $\binom{n}{2}$ terms in the summation, we get
\begin{align*}
P(E_\varepsilon) &\leq \sum_{a,b\in P: (*) \textrm{ holds}} P(E_{a,b})\\
                 &\leq \binom{n}{2}\binom{n-2}{d-1}\left(1-A_d(\delta(a,b))\right)^{(n-d-1)} \leq \varepsilon.
\end{align*}
\end{proof}

The following corollaries for $d=1$ and $d=2$ can be obtained from Theorem~\ref{thm:dsphereub} using the corresponding surface areas.
%(Proof of Corollary~\ref{cor:3sphereub} appears in the Appendix.)

\begin{cor}
In the Delaunay graph $D(P)$ of a set $P$ of $n>2$ points distributed
uniformly and independently at random on a unit circle
($1$-sphere), with probability at least $1-\varepsilon$, for $0<
\varepsilon < 1$, there is no arc $\widehat{ab}\in D(P)$, $a,b\in P$,
such that
$$\delta(a,b) \geq \frac{\ln\left(\binom{n}{2}\big/\varepsilon\right)}{n-2}.$$
\end{cor}

%\begin{proof}
%The surface area of a spherical cap of a $1$-sphere is the length of an arc on a circumference, which in this case is $\delta(a,b)$.
%Replacing in Theorem~\ref{thm:dsphereub}, the claim follows.
%\end{proof}

\begin{cor}
\label{cor:3sphereub}
In the Delaunay graph $D(P)$ of a set $P$ of \mig{$n>3$} points distributed uniformly and independently at random on a unit sphere ($2$-sphere), 
with probability at least $1-\varepsilon$, for $0< \varepsilon < 1$, 
there is no arc $\widehat{ab}\in D(P)$, $a,b\in P$, such that 
$$\delta(a,b) \geq \frac{1}{\sqrt{\pi}} \cos^{-1} \left(1-\frac{2\ln\left(\binom{n}{2}\mig{(n-2)}\big/\varepsilon\right)}{\mig{n-3}}\right).$$
\end{cor}

%\full{
\begin{proof}
The surface area of a spherical cap of a $2$-sphere is $2\pi Rh$, where $R$ is the radius of the sphere and $h$ is the height of the cap.
For a unit $2$-sphere is $R=1/(2\sqrt{\pi})$.
Then, the perimeter of a great circle is $2\pi/(2\sqrt{\pi})=\sqrt{\pi}$.
Thus, the central angle of a cap whose orthodromic diameter is $\rho$ is $2\pi\rho/\sqrt{\pi}=2\sqrt{\pi}\rho$.
Let the angle between the line segment $\overline{ab}$ and the radius of the sphere be $\alpha$.
Then, 
\begin{displaymath} 
\alpha = \left\{ \begin{array}{ll} 
\pi/2-\sqrt{\pi}\rho & \textrm{if $\rho\leq\sqrt{\pi}/2$}\\ 
\sqrt{\pi}\rho-\pi/2 & \textrm{if $\rho>\sqrt{\pi}/2$} 
\end{array} \right. 
\end{displaymath} 

And the height of the cap is $h=1/(2\sqrt{\pi})-1/(2\sqrt{\pi}) \sin (\pi/2-\sqrt{\pi}\rho) = (1-\cos (\sqrt{\pi}\rho) ) / (2\sqrt{\pi})$. 
%\begin{displaymath} 
%h = \left\{ \begin{array}{ll} 
%1/(2\sqrt{\pi})-1/(2\sqrt{\pi}) \sin (\pi/2-\sqrt{\pi}\rho) & \textrm{if $\rho\leq\sqrt{\pi}/2$}\\ 
%1/(2\sqrt{\pi})+1/(2\sqrt{\pi}) \sin (\sqrt{\pi}\rho-\pi/2) & \textrm{if $\rho>\sqrt{\pi}/2$} 
%\end{array} \right. 
%\end{displaymath} 
%\begin{displaymath} 
%h = \left\{ \begin{array}{ll} 
%1/(2\sqrt{\pi})-1/(2\sqrt{\pi}) \sin (\pi/2-\sqrt{\pi}\rho) & \textrm{if $\rho\leq\sqrt{\pi}/2$}\\ 
%1/(2\sqrt{\pi})-1/(2\sqrt{\pi}) \sin (\pi/2-\sqrt{\pi}\rho) & \textrm{if $\rho>\sqrt{\pi}/2$} 
%\end{array} \right. 
%\end{displaymath} 
%\begin{displaymath} 
%h = \left\{ \begin{array}{ll} 
%(1-\cos (\sqrt{\pi}\rho) ) / (2\sqrt{\pi}) & \textrm{if $\rho\leq\sqrt{\pi}/2$}\\ 
%(1+\cos (\sqrt{\pi}\rho) ) / (2\sqrt{\pi}) & \textrm{if $\rho>\sqrt{\pi}/2$} 
%\end{array} \right. 
%\end{displaymath} 
Therefore, the surface area of a spherical cap of a $2$-sphere whose orthodromic diameter is $\rho$ is $(1-\cos (\sqrt{\pi}\rho) ) / 2$.
Replacing in Theorem~\ref{thm:dsphereub}, the claim follows.
\end{proof}
%}

%%%%%%%%%%%%%%%%%%%%%%%%%%%%%%%%%%%%%%%%%%
\subsection{Lower Bound}

\begin{theorem}
\label{thm:dspherelb}
In the Delaunay graph $D(P)$ of a set $P$ of $n>2$ points distributed uniformly and independently at random in a unit $d$-sphere, 
with probability at least $\varepsilon$, 
there is an arc $\widehat{ab}\in D(P)$, $a,b\in P$, such that $A_{d}(\delta(a,b)) \geq A_{d}(\rho_1)$, 
where 
\begin{align*}
A_{d}(\rho_1) &= \frac{\ln\left((e-1)/(e^2\varepsilon)\right)}{n-2+\ln\left((e-1)/(e^2\varepsilon)\right)},
\end{align*}
for any $0< \varepsilon < 1$ such that $A_{d}(2\rho_1)\leq1-1/(n-1)$.
\end{theorem}

\begin{proof}
In order to prove this claim, we consider a configuration given by a specific pair of points and a specific empty spherical cap circumscribing them, that would yield a Delaunay arc between those points. Then, we relate the probability of existence of such configuration to the distance between the points. 
Finally, we relate this quantity to the desired parametric probability.
The details follow.

For any pair of points $a,b\in P$, by Definition~\ref{dsphDel}, for the arc $\widehat{ab}$ to be in $D(P)$, there must exist a $d$-dimensional spherical cap $C$ such that $a$ and $b$ are located on the boundary of the cap base and the cap surface of $C$ is void of points from $P$. 
We compute the probability of such an event as follows.
Let $\rho_2>\rho_1$ be such that $A_{d}(2\rho_2)-A_{d}(2\rho_1)=1/(n-1)$.
Consider any point $a\in P$.
The probability that some other point $b$ is located so that $\rho_1<\delta(a,b)\leq\rho_2$ is
$1-\left(1-1/(n-1)\right)^{n-1} \geq 1-1/e$, by Inequality~\ref{taylor}.
%\begin{align*}
%1-\left(1-\frac{1}{n-1}\right)^{n-1}
%&\geq 1-1/e,\textrm{ by Inequality~\ref{taylor}}.
%\end{align*}

The spherical cap with orthodromic diameter $\delta(a,b)$ is empty with probability
$\left( 1 - A_{d}(\delta(a,b))\right)^{n-2}$.
To lower bound this probability we consider separately the spherical cap with orthodromic diameter $\rho_1$ and the remaining annulus of the spherical cap with orthodromic diameter $\delta(a,b)$. 
The probability that the latter is empty is lower bounded by upper bounding the area $A_{d}(\delta(a,b))-A_{d}(\rho_1)\leq A_{d}(2\rho_2)-A_{d}(2\rho_1)=1/(n-1)$. Then,
$\left(1-1/(n-1)\right)^{n-2} \geq 1/e$, by Inequality~\ref{taylor}.
%\begin{align*}
%\left(1-\frac{1}{n-1}\right)^{n-2} &\geq 1/e,\textrm{ by Inequality~\ref{taylor}}.
%\end{align*}

Finally, the probability that the spherical cap with orthodromic diameter $\rho_1$ is empty is, by Inequality~\ref{taylor},
\begin{align*}
\left(1-A_{d}(\rho_1)\right)^{n-2} &\geq \exp\left(-\frac{A_{d}(\rho_1)(n-2)}{1-A_{d}(\rho_1)}\right),\\
&= \exp\left(-\ln\left(\frac{e-1}{e^2\varepsilon}\right)\right)
= \frac{e^2\varepsilon}{e-1}.
\end{align*}

Therefore,
\begin{align*}
Pr\left(\widehat{ab}\in D(P)\right) &\geq \left(1-\frac{1}{e}\right)\frac{1}{e}\frac{e^2\varepsilon}{e-1} = \varepsilon.
\end{align*}
%\begin{align*}
%\left(1-\frac{1}{e}\right)\frac{1}{e}\left(1-A_{d}(\rho_1)\right)^{n-2} \geq \varepsilon\\
%\left(1-A_{d}(\rho_1)\right)^{n-2} \geq \varepsilon\frac{e^2}{e-1}\\
%\exp\left(-\frac{A_{d}(\rho_1)(n-2)}{1-A_{d}(\rho_1)}\right) \geq \varepsilon\frac{e^2}{e-1}\\
%-\frac{A_{d}(\rho_1)(n-2)}{1-A_{d}(\rho_1)} \geq \ln\left(\varepsilon\frac{e^2}{e-1}\right)\\
%\frac{A_{d}(\rho_1)(n-2)}{1-A_{d}(\rho_1)} \leq \ln\left(\frac{e-1}{e^2\varepsilon}\right)\\
%A_{d}(\rho_1)(n-2) \leq (1-A_{d}(\rho_1)) \ln\left(\frac{e-1}{e^2\varepsilon}\right)\\
%A_{d}(\rho_1)(n-2+\ln\left(\frac{e-1}{e^2\varepsilon}\right)) \leq \ln\left(\frac{e-1}{e^2\varepsilon}\right)\\
%A_{d}(\rho_1) \leq \frac{\ln\left(\frac{e-1}{e^2\varepsilon}\right)}{n-2+\ln\left(\frac{e-1}{e^2\varepsilon}\right)}\\
%\end{align*}
\end{proof}

The following corollaries for $d=1$ and $d=2$ can be obtained from Theorem~\ref{thm:dspherelb} using the corresponding surface areas.

\begin{cor}
In the Delaunay graph $D(P)$ of a set $P$ of $n>2$ points distributed uniformly and independently at random in a unit circle ($1$-sphere), 
with probability at least $\varepsilon$, 
for any $e^{1-n-4/n} \leq \varepsilon < 1$,
there is an arc $\widehat{ab}\in D(P)$, $a,b\in P$, such that 
\begin{align*}
\delta(a,b) &\geq \frac{\ln\left((e-1)/(e^2\varepsilon)\right)}{n-2+\ln\left((e-1)/(e^2\varepsilon)\right)}.
\end{align*}
\end{cor}

%\begin{proof}
%\begin{align*}
%\delta(a,b) \geq \rho_1 &= \frac{\ln\left((e-1)/(e^2\varepsilon)\right)}{n-2+\ln\left((e-1)/(e^2\varepsilon)\right)}\\
%2\rho_1 &\leq 1-1/(n-1)\\
%\frac{\ln\left((e-1)/(e^2\varepsilon)\right)}{n-2+\ln\left((e-1)/(e^2\varepsilon)\right)} &\leq \frac{n-2}{2n-2}\\
%\ln\left(\frac{e-1}{e^2\varepsilon}\right) &\leq \frac{(n-2)^2}{2n-2}+\frac{n-2}{2n-2}\ln\left(\frac{e-1}{e^2\varepsilon}\right)\\
%\ln\left(\frac{e-1}{e^2\varepsilon}\right)\left(1-\frac{n-2}{2n-2}\right) &\leq \frac{(n-2)^2}{2n-2}\\
%\ln\left(\frac{e-1}{e^2\varepsilon}\right) &\leq \frac{(n-2)^2}{n}\\
%\frac{e-1}{e^2\varepsilon} &\leq \exp\left(\frac{(n-2)^2}{n}\right)\\
%\varepsilon &\geq \frac{e-1}{e^2} \exp\left(-\frac{(n-2)^2}{n}\right)\\
%\varepsilon &\geq (e-1) \exp\left(-\frac{n^2-2n+4}{n}-2\right)\\
%\varepsilon &\geq (e-1) \exp\left(-n-4/n\right)\\
%\varepsilon &\geq \exp\left(1-n-4/n\right).
%\end{align*}
%\end{proof}

\begin{cor}
In the Delaunay graph $D(P)$ of a set $P$ of $n>2$ points distributed uniformly and independently at random in a unit sphere ($2$-sphere), 
with probability at least $\varepsilon$, 
for any $e^{-n+2\sqrt{n-1}-1}\leq \varepsilon < 1$,
there is an arc $\widehat{ab}\in D(P)$, $a,b\in P$, such that 
\begin{align*}
\delta(a,b) \geq \frac{1}{\sqrt{\pi}}\cos^{-1}\left(1-\frac{2\ln\left((e-1)/(e^2\varepsilon)\right)}{n-2+\ln\left((e-1)/(e^2\varepsilon)\right)}\right).
\end{align*}
\end{cor}

\begin{proof}
As shown in the proof of Corollary~\ref{cor:3sphereub}, the surface area of a spherical cap of a $2$-sphere whose orthodromic diameter is $\rho$ is $(1-\cos (\sqrt{\pi}\rho) ) / 2$.
Replacing in Theorem~\ref{thm:dspherelb}, the claim follows.
\end{proof}

\section{Enclosing Body with Boundary}
\label{section:wb}

The following theorems show upper and lower bounds on the length of edges in the Delaunay graph in a $d$-ball. 

%%%%%%%%%%%%%%%%%%%%%%%%%%%%%%%%%%%%%
\subsection{Upper Bound}

%In the Appendix we prove the following theorem:

\begin{theorem}
\label{thm:dballub}
In the Delaunay graph $D(P)$ of a set $P$ of \mig{$n>d+1\geq 2$} points distributed uniformly and independently at random in a unit $d$-ball, with probability at least $1-\varepsilon$, for $0< \varepsilon < 1$, 
there is no edge $(a,b)\in D(P)$, $a,b\in P$, such that 
$$V_{d}(d(a,b))\geq \frac{\ln\left(\binom{n}{2}\mig{\binom{n-2}{d-1}}\big/\varepsilon\right)}{\mig{n-d-1}}.$$
\end{theorem}

%\full{
\begin{proof}

\mig{In order to prove this claim, we consider any one set of $d+1$ points in $P$.
Then, we relate the probability that the ball circumscribing the set is empty, to the distance separating the points.
Finally, we combine the probabilities for all possible pairs of points and sets and we relate this quantity to the desired parametric probability.
The details follow.
}

%\mig{
%Consider first any set $S\subseteq P$ of $d+1$ points in general position. 
%Let $B$ be the $d$-dimensional ball circumscribing the set $S$, which is unique because the points are in general position. 
%For any $a,b\in S$, by Definition~\ref{dsphDel}, if the interior of the ball is void of points, then $\widehat{ab}\in D(P)$.
%}
\mig{
Consider any pair of points $a,b\in P$, and any set $S\subseteq P$ of $d+1$ points such that $a,b\in S$. 
Let $B$ be the $d$-dimensional ball circumscribing the set $S$, which is unique because the points are in general position. 
By Definition~\ref{dsphDel}, if the interior of $B$ is void of points, then $\widehat{ab}\in D(P)$.
}

%
%
%Let $a,b\in P$ be any pair of points from $P$.
%By Definition~\ref{dballDel}, for the edge $(a,b)$ to be in $D(P)$, there must exist a $d$-dimensional ball $B$ such that $a$ and $b$ are located on the surface area of $B$ and the interior of $B$ is void of points from $P$. 
Notice that $B$ may be such that part of it is outside the unit ball but, given that points are distributed in the unit ball, no point is located outside of it.
Then, in order to upper bound $Pr((a,b)\in D(P))$, we upper bound the probability that the interior of the intersection of $B$ with the unit ball is empty. 
In order to do that, we lower bound the volume of such an intersection.
For a given distance $d(a,b)$ such volume is minimized when $a$ and $b$ are located on the surface of the unit ball, $B$ has infinite radius, and the maximum distance between any pair of points in the intersection between the surface areas of $B$ and the unit ball is $d(a,b)$ (i.e., the distance is maximized when the intersection is a spherical cap of the unit ball with base diameter $d(a,b)$). 
Therefore, given that the points are distributed uniformly and independently at random, it is 
$Pr((a,b)\in D(P))\leq(1-V_{d}(d(a,b)))^{n-d-1}$.

%\mig{
%We now turn to our attention to any degenerate set $S\subseteq P$ of $d+1$ points. 
%Consider any pair $a,b\in S$ such that, together with \emph{any} other $d-1$ points, are not in general position.
%By Definition~\ref{dsphDel}, if there is \emph{some} ball $B$ such that $a$ and $b$ are located in the boundary and the interior of $B$ is empty, then $\widehat{ab}\in D(P)$. Checking this property would yield infinitely many balls to check. However, given that the pair $a,b$ does not belong to \emph{any} generic set of size $(d+1)$, then \emph{all} points in $P-S$ are located in the same hyperplane of $d-1$ dimensions that the points in $S$ are. Hence, we know that both open half-spaces separated by such hyperplane are in fact empty. Thus, the above bound holds also to this case.
%}
%
%\mig{
%It remains to consider any pair $a,b \in S$ that is contained in some generic set of $d+1$ points. 
%It can be seen that in this case the argument is a combination of the above. Either one half space is empty or both half spaces have other points, in both cases the bound above applies.
%}

Given that there are $\binom{n}{2}$ pairs of points, \mig{and $\binom{n-2}{d-1}$ sets of $d+1$ points including a given pair},  using the union bound, we find a $V_{d}(d(a,b))$ that yields a probability at most $\varepsilon$ of having some edge $(a,b)\in D(P)$, by making
\begin{align*}
\binom{n}{2}\mig{\binom{n-2}{d-1}}\left(1-V_{d}(d(a,b))\right)^{\mig{n-d-1}} &\leq \varepsilon.
\end{align*}
Then, given that $a\neq b$, it holds that $V_{d}(d(a,b))<1$. Then, using Inequality~\ref{taylor}, it is enough that
\begin{align*}
\exp\left(-V_{d}(d(a,b))(\mig{n-d-1})\right) &\leq \varepsilon\big/\left(\binom{n}{2}\mig{\binom{n-2}{d-1}}\right),
\end{align*}
which is implied by
\begin{align*}
V_{d}(d(a,b)) \geq \frac{\ln\left(\binom{n}{2}\mig{\binom{n-2}{d-1}}\big/\varepsilon\right)}{\mig{n-d-1}}.
\end{align*}
\end{proof}
%}

The following corollaries for $d=2$ and $d=3$ can be obtained from Theorem~\ref{thm:dballub} using the corresponding volumes. 
%The proofs are in the Appendix.

\begin{cor}
\label{cor:d2}
In the Delaunay graph $D(P)$ of a set $P$ of \mig{$n>3$} points distributed uniformly and independently at random in a unit disk ($2$-ball), with probability at least $1-\varepsilon$, for $\mig{\binom{n}{2}(n-2)e^{-\sqrt{2}(n-3)/\pi}} < \varepsilon < 1$, 
there is no edge $(a,b)\in D(P)$, $a,b\in P$, such that 
%$$V_{2}(d(a,b)) \geq \frac{\ln\left(\mig{\binom{n}{3}}\big/\varepsilon\right)}{\mig{n-3}}.$$
$$d(a,b) \geq \sqrt[3]{\frac{16}{\sqrt{\pi}}\frac{\ln\left(\binom{n}{2}\mig{(n-2)}\big/\varepsilon\right)}{\mig{n-3}}}.$$
\end{cor}

%\full{
\begin{proof}
Consider the intersection of the radius of the unit disk perpendicular to $(a,b)$ with the circumference of the unit disk, call this point $d$. The area of the triangle $\bigtriangleup abd$ is a strict lower bound on $V_{2}(d(a,b))$. From Theorem~\ref{thm:dballub}, we have the condition
\begin{align*}
V_{2}(d(a,b))\geq \frac{\ln\left(\binom{n}{2}\mig{(n-2)}\big/\varepsilon\right)}{\mig{n-3}}.
\end{align*}

Thus, it is enough
\begin{align*}
\frac{d(a,b)}{2}\left(\frac{1}{\sqrt{\pi}}-\sqrt{\frac{1}{\pi}-\frac{d(a,b)^2}{4}}\right) &\geq \frac{\ln\left(\binom{n}{2}\mig{(n-2)}\big/\varepsilon\right)}{\mig{n-3}}.
\end{align*}

Making $\rho=d(a,b)\sqrt{\pi}/2$, we want
\begin{align}
\sqrt{\rho^2-\rho^4} &\leq \rho-\pi\frac{\ln\left(\binom{n}{2}\mig{(n-2)}\big/\varepsilon\right)}{\mig{n-3}}.\label{eqrho}
\end{align}
If $d(a,b)< 2\sqrt{\pi}\ln\left(\binom{n}{2}\mig{(n-2)}\big/\varepsilon\right)\big/(\mig{n-3})$, there is nothing to prove because 
\begin{align*}
\frac{2\sqrt{\pi}\ln\left(\binom{n}{2}\mig{(n-2)}\big/\varepsilon\right)}{\mig{n-3}} < \sqrt[3]{\frac{16\ln\left(\binom{n}{2}\mig{(n-2)}\big/\varepsilon\right)}{\sqrt{\pi}(\mig{n-3})}},
\end{align*}

for any $\varepsilon > \binom{n}{2}\mig{(n-2)}\exp\left(-\sqrt{2}(\mig{n-3})/\pi\right)$. Then, we can square both sides of Inequality~\ref{eqrho} getting 
\begin{align*}
\rho^4 &\geq 2\rho\pi\frac{\ln\left(\binom{n}{2}\mig{(n-2)}\big/\varepsilon\right)}{\mig{n-3}}-\left(\pi\frac{\ln\left(\binom{n}{2}\mig{(n-2)}\big/\varepsilon\right)}{\mig{n-3}}\right)^2,
\end{align*}
which is implied by
\begin{align*}
\rho^3 \geq 2\pi\frac{\ln\left(\binom{n}{2}\mig{(n-2)}\big/\varepsilon\right)}{\mig{n-3}}.
\end{align*}
And replacing back $\rho$, the claim follows.
\end{proof}
%}

\begin{cor}
\label{cor:d3}
In the Delaunay graph $D(P)$ of a set $P$ of \mig{$n>4$} points distributed uniformly and independently at random in a unit ball ($3$-ball), with probability at least $1-\varepsilon$, for $\mig{\binom{n}{2}\binom{n-2}{2}e^{-2(n-4)\big/(3\sqrt{\pi})}}< \varepsilon < 1$, 
there is no edge $(a,b)\in D(P)$, $a,b\in P$, such that 
$$d(a,b) \geq \sqrt[4]{\frac{96}{\pi^{3/2}}\frac{\ln\left(\binom{n}{2}\mig{\binom{n-2}{2}}\big/\varepsilon\right)}{\mig{n-4}}}.$$
\end{cor}

%\full{
\begin{proof}
Consider the intersection of the radius of the unit ball perpendicular to $(a,b)$ with the surface of the unit ball, call this point $d$. The volume of the cone whose base is the disk whose diameter is $(a,b)$ and its vertex is $d$ is a strict lower bound on $V_{2}(d(a,b))$. From Theorem~\ref{thm:dballub}, we have the condition
\begin{align*}
V_{3}(d(a,b))\geq \frac{\ln\left(\binom{n}{2}\mig{\binom{n-2}{2}}\big/\varepsilon\right)}{\mig{n-4}}.
\end{align*}

Thus, it is enough
\begin{align*}
\frac{\pi}{3} \left(\frac{d(a,b)}{2}\right)^2 \left(\frac{1}{\sqrt{\pi}}-\sqrt{\frac{1}{\pi}-\frac{d(a,b)^2}{4}}\right) &\geq \frac{\ln\left(\binom{n}{2}\mig{\binom{n-2}{2}}\big/\varepsilon\right)}{\mig{n-4}}.
\end{align*}

Making $\rho=d(a,b)\sqrt{\pi}/2$, we want
\begin{align}
\sqrt{\rho^4-\rho^6} &\leq \rho^2-3\sqrt{\pi}\frac{\ln\left(\binom{n}{2}\mig{\binom{n-2}{2}}\big/\varepsilon\right)}{\mig{n-4}}.\label{eqrho2}
\end{align}

If $d(a,b)<\sqrt{12\ln\left(\binom{n}{2}\mig{\binom{n-2}{2}}\big/\varepsilon\right)\big/(\sqrt{\pi}(\mig{n-4}))}$, there is nothing to prove because 
\begin{align*}
\sqrt{\frac{12\ln\left(\binom{n}{2}\mig{\binom{n-2}{2}}\big/\varepsilon\right)}{\sqrt{\pi}(\mig{n-4})}} < \sqrt[4]{\frac{96}{\pi^{3/2}}\frac{\ln\left(\binom{n}{2}\mig{\binom{n-2}{2}}\big/\varepsilon\right)}{\mig{n-4}}},
\end{align*}

for any $\varepsilon > \binom{n}{2}\mig{\binom{n-2}{2}}\exp\left(-2(\mig{n-4})\big/(3\sqrt{\pi})\right)$. Then, we can square both sides of Inequality~\ref{eqrho2}, getting 
\begin{align*}
\rho^6 &\geq 6\rho^2\sqrt{\pi}\frac{\ln\left(\binom{n}{2}\mig{\binom{n-2}{2}}\big/\varepsilon\right)}{\mig{n-4}}-\left(3\sqrt{\pi}\frac{\ln\left(\binom{n}{2}\mig{\binom{n-2}{2}}\big/\varepsilon\right)}{\mig{n-4}}\right)^2
\end{align*}
which is implied by
\begin{align*}
\rho^4 \geq 6\sqrt{\pi}\frac{\ln\left(\binom{n}{2}\mig{\binom{n-2}{2}}\big/\varepsilon\right)}{\mig{n-4}}.
\end{align*}
And replacing back $\rho,$ the claim follows.
\end{proof}
%}

%%%%%%%%%%%%%%%%%%%%%%%%%%%%%%%%%%%%%%%%%%
\subsection{Lower Bound}

In this section we give lower bounds for $d=2$ and $d=3$. As in the
case without boundary, we prove our lower bounds showing a
configuration given by a specific pair of points and a specific empty
body circumscribing them, that would yield a Delaunay edge between
those points. Then, we relate the probability of existence of such
configuration to the distance between the points and to the desired
parametric probability.

%The following theorems are proved in the Appendix. 

\begin{theorem}
\label{thm:2dballlb}
For $d=2$,
given the Delaunay graph $D(P)$ of a set $P$ of $n>2$ points distributed uniformly and independently at random in a unit $d$-ball, 
with probability at least $\varepsilon$, 
there is an edge $(a,b)\in D(P)$, $a,b\in P$, such that $d(a,b) \geq \rho_1/2$, 
where 
$$V_{d}(\rho_1) = \frac{\ln\left(\alpha/\varepsilon\right)}{\left(n-2+\ln\left(\alpha/\varepsilon\right)\right)},$$
where 
$\alpha=(1-e^{-1/16}) (1-e^{-1/32})e^{-1}$,
for any $0< \varepsilon \leq \alpha/e^2$ such that $V_{d}(\rho_1)\leq1/2-1/n$. 
\mig{
Which implies that 
$$d(a,b) \geq \sqrt[3]{ \frac{\ln\left(\alpha/\varepsilon\right)}{2\sqrt{\pi}\left(n-2+\ln\left(\alpha/\varepsilon\right)\right)}}.$$
}
\end{theorem}

%\full{
\begin{proof}
For any pair of points $a,b\in P$, by Definition~\ref{dballDel}, for the edge $(a,b)$ to be in $D(P)$, there must exist a $d$-ball such that $a$ and $b$ are located in the surface area of the ball and the interior is void of points from $P$. 
We compute the probability of such event as follows.
(Refer to Figure~\ref{fig:sphcaps}.)
Consider two spherical caps of the unit ball with concentric surface areas, call them $\Gamma_1$ and $\Gamma_2$, of diameters $\rho_1$ as defined and $\rho_2$ such that $V_{d}(\rho_2)=V_{d}(\rho_1)+1/n$. Let $\Gamma_2-\Gamma_1$ be all space points in $\Gamma_2$ that are not in $\Gamma_1$ (i.e. the body defined by the difference of both spherical caps).
Inside $\Gamma_2-\Gamma_1$ (see Figure~\ref{fig:sphcaps1}) consider two bodies $B_a$ and $B_b$ of identical volumes such that for any pair of points $a\in B_a$ and $b\in B_b$ the following holds: \emph{(i)} the points $a$ and $b$ are separated a distance at least $\rho_1/2$; \emph{(ii)} there exists a spherical cap $\Gamma$ containing the points $a$ and $b$ in its base of diameter $\rho$ such that $V_d(\rho)\leq V_{d}(\rho_2)$. (See Figure~\ref{fig:sphcaps2}.)
Such event implies the existence of an empty $d$-ball of infinite radius with $a$ and $b$ in its surface which proves the claim.
In the following, we show that such event occurs with big enough probability.

To bound the volume of $B_a$ (hence, $B_b$), we first bound the ratio $\rho_2/\rho_1$. Consider the inscribed polygons illustrated in Figure~\ref{fig:sphcaps3}. It can be seen that the triangle $x_1x_3x_5$ is located inside the pentagon $x_1x_2x_3x_4x_5$ which in turn is composed by the triangle $x_2x_3x_4$ and the trapezoid $x_1x_2x_4x_5$. Then,
\begin{align}
\frac{(h_1+h)\rho_2}{2} &\leq \frac{\rho_1h_1}{2}+\frac{(\rho_1+\rho_2)h}{2} 
\Longleftarrow
h_1\rho_2 \leq (h_1+h)\rho_1.\label{ratio1}
\end{align}

Given that $\varepsilon \leq \alpha/e^2$, we know that $V_{d}(\rho_1) \geq 2/n$. Then, it holds that $h\leq h_1$. Replacing in Equation~\ref{ratio1} we obtain $\rho_2\leq 2\rho_1$.

Then, the volume of $B_a$ is (see Figure~\ref{fig:sphcaps4})
\begin{align*}
\frac{1}{2n}-\frac{\rho_1}{4}h-\left(\frac{\rho_2}{2}-\frac{\rho_1}{4}\right)\frac{h}{2} &=
\frac{1}{2n}-\frac{\rho_1+\rho_2}{2}\frac{h}{2}+\frac{\rho_1}{8}h.
\end{align*}

Using that the volume of the trapezoid $x_1x_2x_4x_5$ is $(\rho_2+\rho_1)h/2\leq 1/n$, the right-hand side of the latter equation is at least $\rho_1h/8$ which, using that $\rho_1\geq \rho_2/2$, is at least $\rho_2h/16\geq 1/(16n)$.

\begin{figure}[t]
\begin{center}
\psfrag{Ba}{$B_a$}
\psfrag{Bb}{$B_b$}
\psfrag{r1}{$\rho_1$}
\psfrag{r1/2}{$\rho_1/2$}
\psfrag{r2}{$\rho_2$}
\psfrag{a}{$a$}
\psfrag{b}{$b$}
\psfrag{1}{$x_1$}
\psfrag{2}{$x_2$}
\psfrag{3}{$x_3$}
\psfrag{4}{$x_4$}
\psfrag{5}{$x_5$}
\psfrag{r1}{$\rho_1$}
\psfrag{r1/2}{$\rho_1/2$}
\psfrag{r1/4}{$\rho_1/4$}
\psfrag{r2/2}{$\rho_2/2$}
\psfrag{h}{$h$}
\psfrag{h1}{$h_1$}
\subfigure[]{
\includegraphics[width=0.4\textwidth]{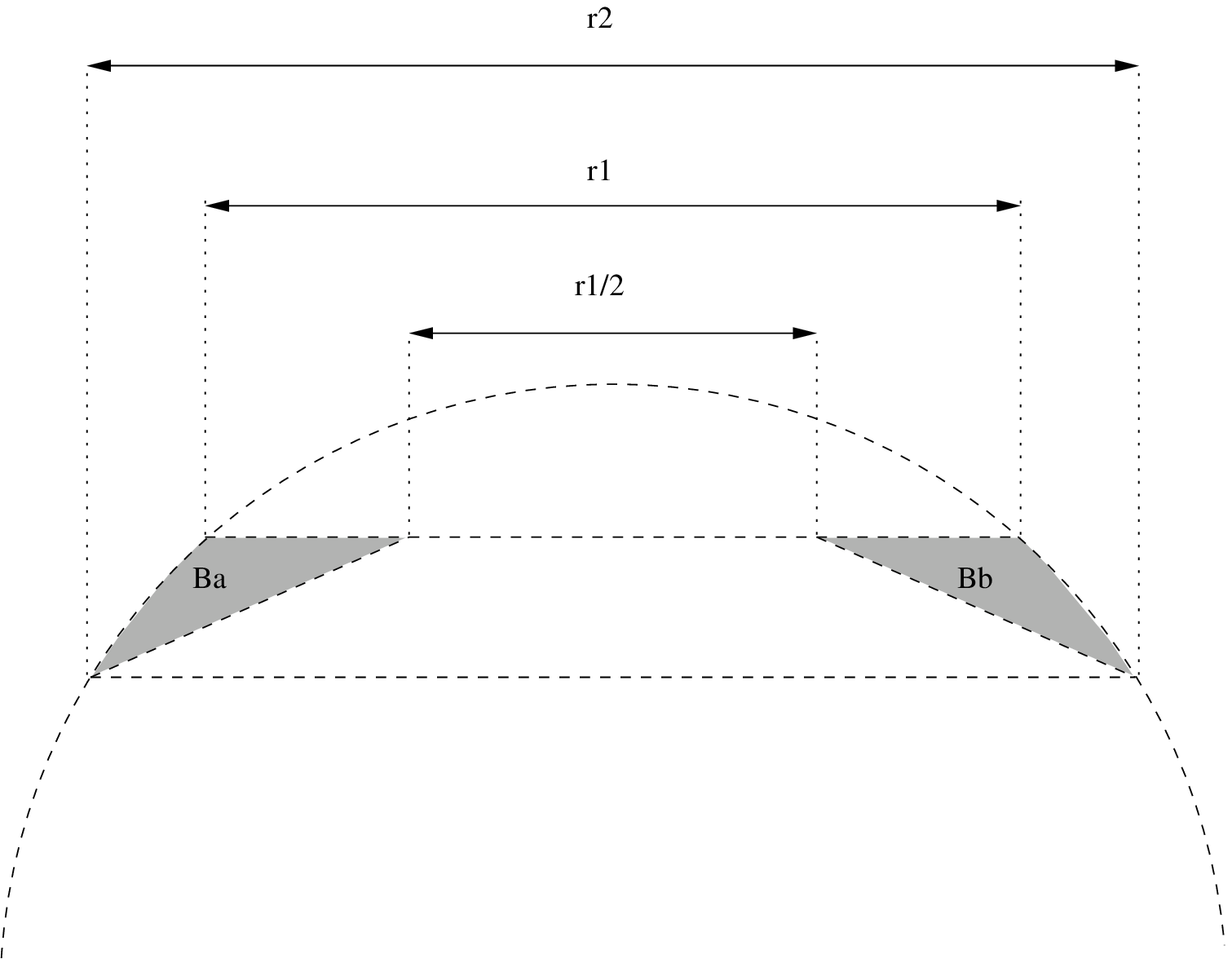}
\label{fig:sphcaps1}
}
\subfigure[]{
\includegraphics[width=0.4\textwidth]{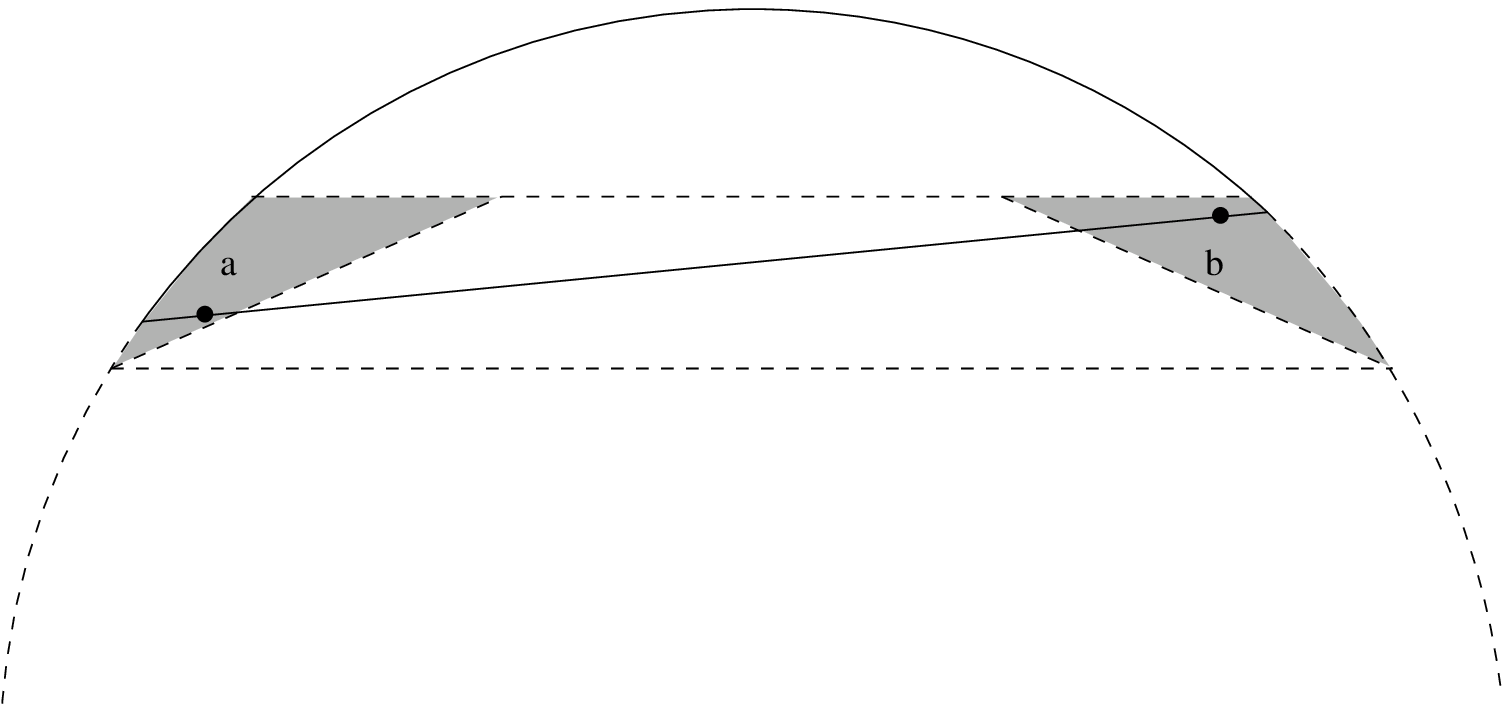}
\label{fig:sphcaps2}
}
\subfigure[]{
\includegraphics[width=0.45\textwidth]{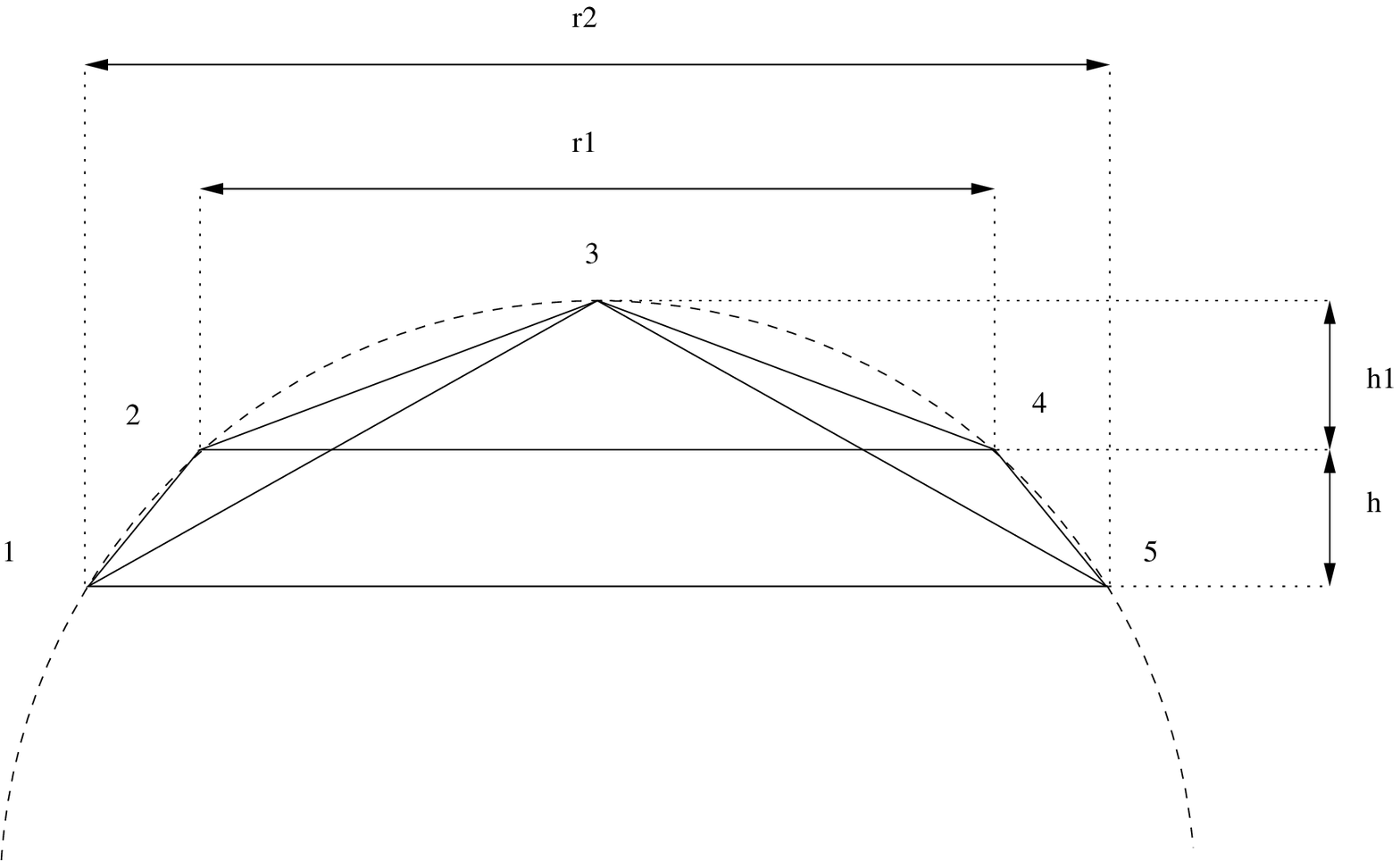}
\label{fig:sphcaps3}
}
\subfigure[]{
\includegraphics[width=0.45\textwidth]{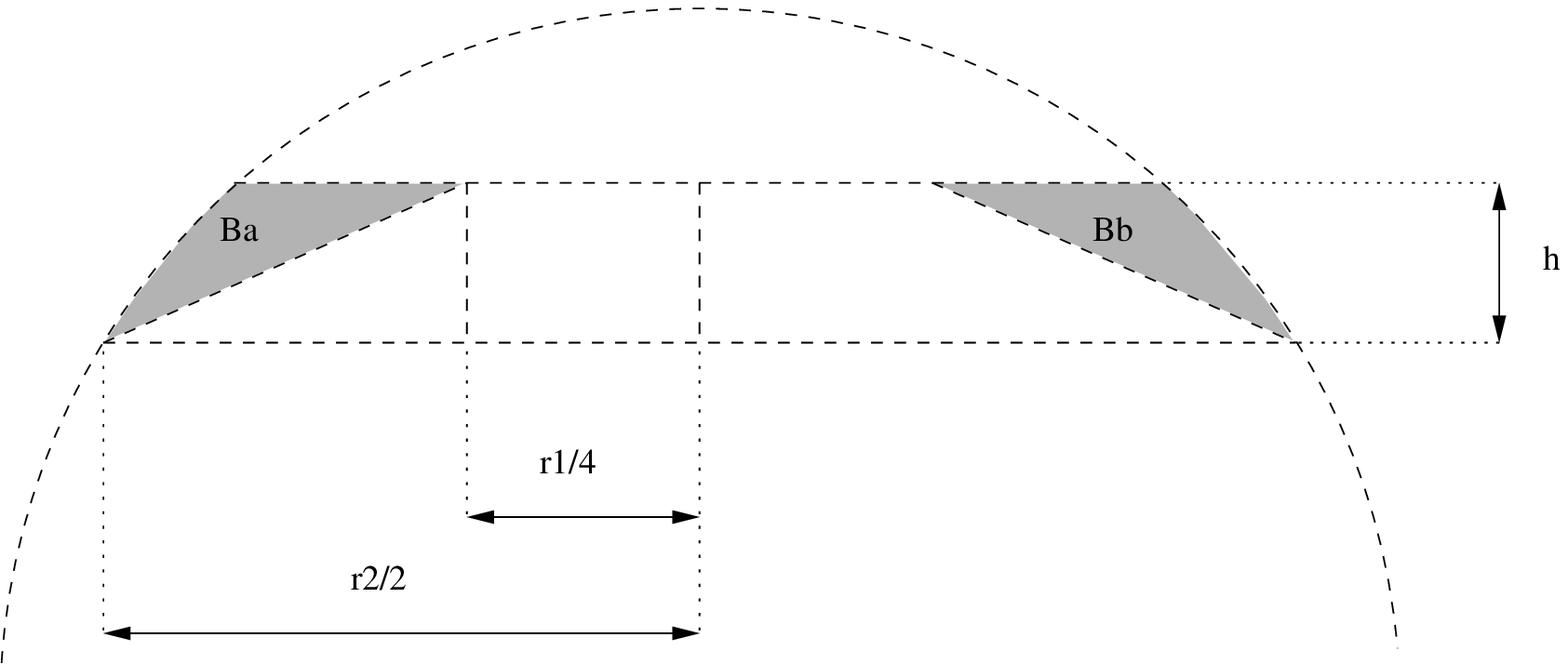}
\label{fig:sphcaps4}
}
\caption{Illustration of Theorem~\ref{thm:2dballlb}.}
\label{fig:sphcaps}
\end{center}
\end{figure}

Then, the probability that there is a point $a\in P$ located in $B_a$ is larger than $1-\left(1-\frac{1}{16n}\right)^{n} 
\geq 1-e^{-1/16}$, by Inequality~\ref{taylor}.
%\begin{align*}
%> 1-\left(1-\frac{1}{16n}\right)^{n}
%&\geq 1-e^{-1/16},\textrm{ by Inequality~\ref{taylor}}.
%\end{align*}
%
And the probability that there is another point $b\in P$ located in $B_b$ is, for any $n>1$,
\begin{align*}
1-\left(1-\frac{1}{16n}\right)^{n-1}
&\geq 1-e^{-(n-1)/(16n)}
\geq 1-e^{-1/32}.
\end{align*}
(The first inequality follows from Inequality~\ref{taylor}.)
It remains to be shown that $\Gamma$ is void of points.
The probability that $\Gamma$ is empty is lower bounded by upper bounding the volume, i.e. taking $V_{d}(\rho) \leq V_{d}(\rho_2)=V_{d}(\rho_1)+(V_{d}(\rho_2)-V_{d}(\rho_1))$. For $V_{d}(\rho_2)-V_{d}(\rho_1)=1/n$, we have
\begin{align*}
\left(1-\frac{1}{n}\right)^{n-2}
&\geq e^{-(n-2)/(n-1)}
\geq 1/e
\end{align*}
(the first inequality follows from Inequality~\ref{taylor}),
and the probability that $\Gamma_1$ is empty is 
\begin{align*}
\left(1-V_{d}(\rho_1)\right)^{n-2} &\geq \exp\left(-\frac{V_{d}(\rho_1)(n-2)}{1-V_{d}(\rho_1)}\right),\textrm{ by Inequality~\ref{taylor}}.
\end{align*}

Replacing, we get
\begin{align*}
Pr\left((a,b)\in D(P)\right) \geq \\
\left(1-\frac{1}{e^{1/16}}\right) \left(1-\frac{1}{e^{1/32}}\right)\frac{1}{e}  \exp\left(-\frac{V_{d}(\rho_1)(n-2)}{1-V_{d}(\rho_1)}\right)
 = \varepsilon.
\end{align*}

\mig{
Which proves the first part of the claim.
For the second part, we upper bound the area of the circular segment of chord $\rho_1$ with the area of the rectangle circumscribing it.
\begin{align*}
V_{2}(\rho_1) &\leq \rho_1\left(\frac{1}{\sqrt{\pi}}-\sqrt{\frac{1}{\pi}-\frac{\rho_1^2}{4}}\right).
\end{align*}
Hence,
\begin{align*}
\sqrt{\frac{\rho_1^2}{\pi}-\frac{\rho_1^4}{4}} &\leq \frac{\rho_1}{\sqrt{\pi}}-V_{2}(\rho_1).
\end{align*}
Given that $\rho_1/\sqrt{\pi}\geq V_{2}(\rho_1)$, we can square both sides getting
\begin{align*}
%\frac{\rho_1^4}{4} &\geq 2\frac{\rho_1}{\sqrt{\pi}}V_{2}(\rho_1)-V_{2}(\rho_1)^2\\
\rho_1^4 &\geq 4\left(2\frac{\rho_1}{\sqrt{\pi}}-V_{2}(\rho_1)\right)V_{2}(\rho_1)\\
&\geq 4\frac{\rho_1}{\sqrt{\pi}}V_{2}(\rho_1), \textrm{ because $V_{2}(\rho_1)\leq\rho_1/\sqrt{\pi}$.}
\end{align*}
Then we get $\rho_1/2 \geq \sqrt[3]{V_{2}(\rho_1)/(2\sqrt{\pi})}$ and replacing $V_{2}(\rho_1)$ the claim follows.
}
\end{proof}
%}

%%%%%%%%%%%%%%%%%%%%%%%%%%%%%%%%%%%%%%%%%%%%%

\begin{theorem}
\label{thm:3dballlb}
For $d=3$,
given the Delaunay graph $D(P)$ of a set $P$ of $n>4$ points distributed uniformly and independently at random in a unit $d$-ball, 
with probability at least $\varepsilon$, 
there is an edge $(a,b)\in D(P)$, $a,b\in P$, such that $d(a,b) \geq \rho_1/2$, 
where 
$$V_{d}(\rho_1) = \frac{\ln\left(\alpha/\varepsilon\right)}{\left(n-2+\ln\left(\alpha/\varepsilon\right)\right)},$$
where 
$\alpha=(1-e^{-1/6}) (1-e^{-1/12})e^{-12}$,
for any $0< \varepsilon \leq \alpha/e$ such that $V_{d}(\rho_1)\leq1/2-1/n$. 
\mig{
Which implies that 
$$d(a,b) \geq \sqrt[4]{ \sqrt[3]{\frac{48}{\pi^4}} \frac{\ln\left(\alpha/\varepsilon\right)}{\left(n-2+\ln\left(\alpha/\varepsilon\right)\right)}}.$$
}
\end{theorem}

\begin{proof}
For any pair of points $a,b\in P$, by Definition~\ref{dballDel}, for the edge $(a,b)$ to be in $D(P)$, there must exist a $d$-ball such that $a$ and $b$ are located in the surface area of the ball and the interior is void of points from $P$. 
We compute the probability of such event as follows.
(Refer to the two-dimensional projections of Figure~\ref{fig:sphcaps3d}.)

\begin{figure}[t]
\begin{center}
\psfrag{Ba}{$B_a$}
\psfrag{Bb}{$B_b$}
\psfrag{a}{$a$}
\psfrag{b}{$b$}
\psfrag{g}{$\Gamma$}
\psfrag{g1}{$\Gamma_1$}
\psfrag{g21}{$\Gamma_2-\Gamma_1$}
\psfrag{1}{$x_1$}
\psfrag{2}{$x_2$}
\psfrag{3}{$x_3$}
\psfrag{4}{$x_4$}
\psfrag{5}{$x_5$}
\psfrag{r1}{$\rho_1$}
\psfrag{r1/sqrt2}{$\rho_1/\sqrt{2}$}
\psfrag{r1/2}{$\rho_1/2$}
\psfrag{r1/4}{$\rho_1/4$}
\psfrag{r2}{$\rho_2$}
\psfrag{r2/2}{$\rho_2/2$}
\psfrag{h}{$h$}
\psfrag{h1}{$h_1$}
\psfrag{h2}{$h_2$}
\subfigure[]{
\includegraphics[width=0.4\textwidth]{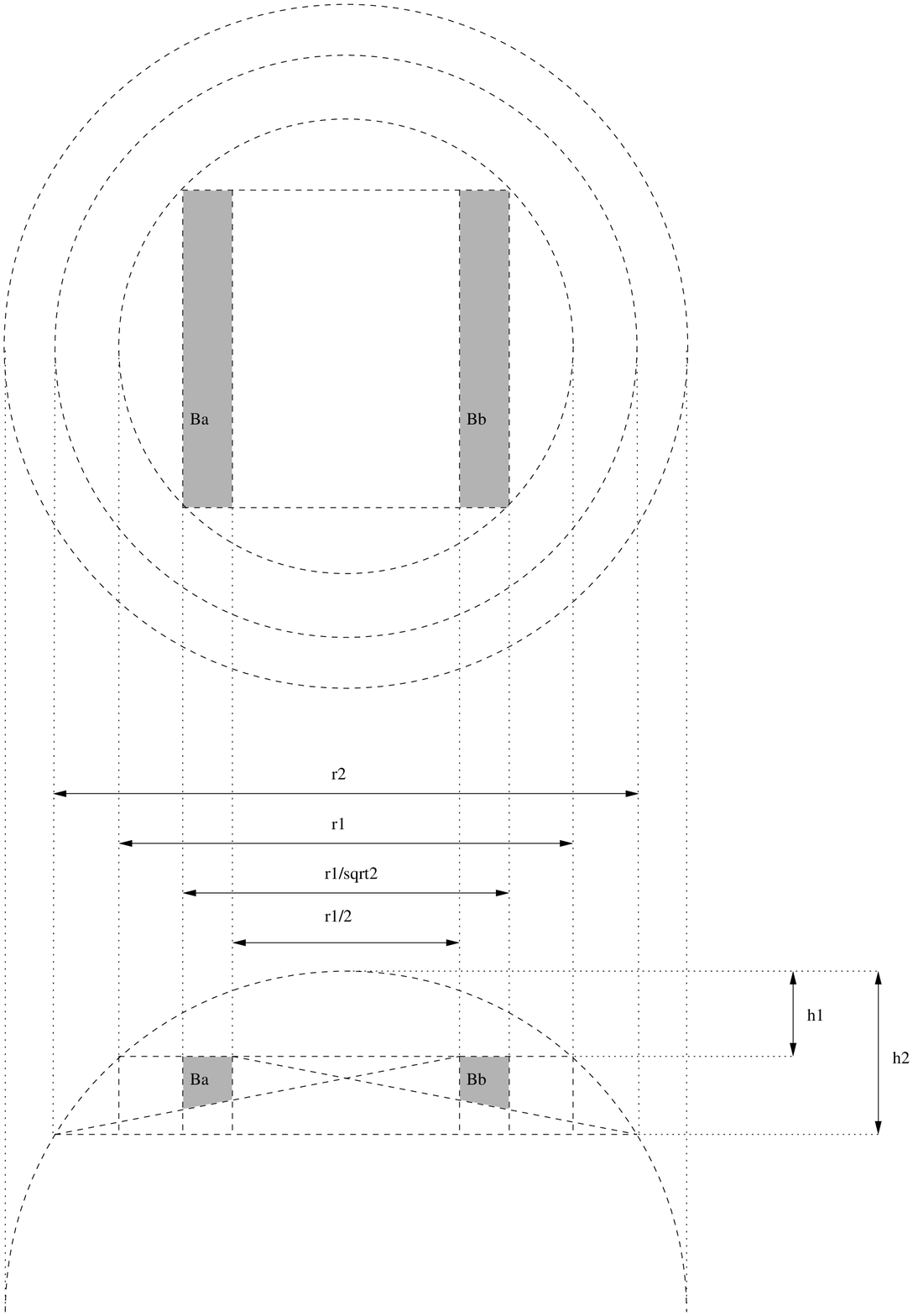}
\label{fig:sphcaps3d_1}
}
\subfigure[]{
\includegraphics[width=0.3\textwidth]{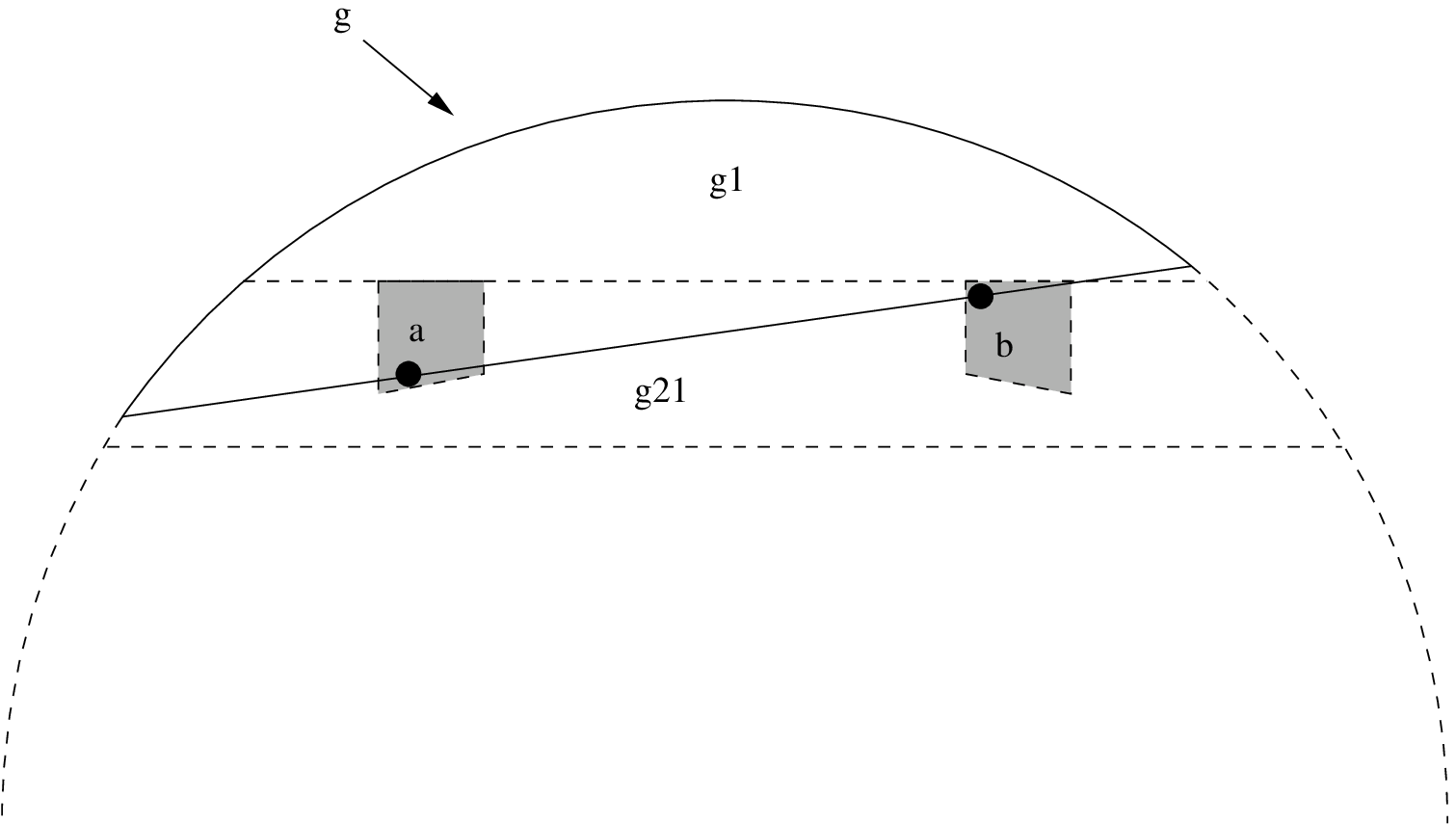}
\label{fig:sphcaps3d_2}
}
\subfigure[]{
\includegraphics[width=0.3\textwidth]{sphcaps_3.eps}
\label{fig:sphcaps3d_3}
}
\caption{Illustration of Theorem~\ref{thm:3dballlb}.}
\label{fig:sphcaps3d}
\end{center}
\end{figure}

Consider two spherical caps of the unit ball with concentric surface areas, call them $\Gamma_1$ and $\Gamma_2$, of base diameters $\rho_1$ and $\rho_2$, and heights $h_1$ and $h_2$ respectively.
Let $\rho_1$ be such that $V_{d}(\rho_1)$ is as defined and let $h_2$ be such that $\pi(\rho_1/2)^2(h_2-h_1)=1/n$.
Let $\Gamma_2-\Gamma_1$ be all space points in $\Gamma_2$ that are not in $\Gamma_1$ (i.e., the body defined by the difference of both spherical caps).
Consider the parallelepiped of sides $\rho_1/\sqrt{2}\times\rho_1/\sqrt{2}\times h_2-h_1$ inscribed in $\Gamma_2-\Gamma_1$ (see Figure~\ref{fig:sphcaps3d_1}), call it $\Pi$.

Inside $\Pi$, consider two bodies $B_a$ and $B_b$ of identical volumes such that for any pair of points $a\in B_a$ and $b\in B_b$ the following holds: \emph{(i)} the points $a$ and $b$ are separated a distance at least $\rho_1/2$; \emph{(ii)} there exists an empty spherical cap $\Gamma$ that contains the points $a$ and $b$ in its base of diameter $\rho$ such that $V_d(\rho)\leq V_{d}(\rho_2)$. (See Figure~\ref{fig:sphcaps3d_2}.)
Such configuration implies the existence of an empty $d$-ball of infinite radius with $a$ and $b$ in its surface which proves the claim.
In the following, we show that such configuration occurs with big enough probability.

To bound the volume of $B_a$ (hence, $B_b$), we first bound the ratio $\rho_2/\rho_1$. Consider the inscribed bodies whose projection is illustrated in Figure~\ref{fig:sphcaps3d_3}. It can be seen that the cone $x_1x_3x_5$ is located inside the body composed by the cone $x_2x_3x_4$ and the frustum $x_1x_2x_4x_5$. Then,
\begin{align}
\frac{h_2\pi(\rho_2/2)^2}{d} &\leq \frac{h_1\pi(\rho_1/2)^2}{d}+\frac{\pi(\rho_2/2)^3-\pi(\rho_1/2)^3}{d(\rho_2/2-\rho_1/2)}(h_2-h_1)\label{ratio2},
\end{align}
which is implied by $h_1\rho_2 \leq h_2\rho_1$.
%
%
%h_2(\rho_2/2)^2 &\leq h_1(\rho_1/2)^2+\frac{(\rho_2/2)^3-(\rho_1/2)^3}{\rho_2/2-\rho_1/2}(h_2-h_1)\\
%h_2(\rho_2/2)^2(\rho_2/2-\rho_1/2) &\leq h_1(\rho_1/2)^2(\rho_2/2-\rho_1/2)+((\rho_2/2)^3-(\rho_1/2)^3)(h_2-h_1)\\
%(h_2(\rho_2/2)^2\rho_2/2-h_2(\rho_2/2)^2\rho_1/2) &\leq (h_1(\rho_1/2)^2\rho_2/2-h_1(\rho_1/2)^2\rho_1/2)\\&+(h_2((\rho_2/2)^3-(\rho_1/2)^3)-h_1((\rho_2/2)^3-(\rho_1/2)^3))\\
%h_2(\rho_2/2)^2\rho_2/2-h_2(\rho_2/2)^2\rho_1/2 &\leq h_1(\rho_1/2)^2\rho_2/2-h_1(\rho_1/2)^2\rho_1/2\\&+(h_2(\rho_2/2)^3-h_2(\rho_1/2)^3)-(h_1(\rho_2/2)^3-h_1(\rho_1/2)^3)\\
%h_2(\rho_2/2)^2\rho_2/2-h_2(\rho_2/2)^2\rho_1/2 &\leq h_1(\rho_1/2)^2\rho_2/2-h_1(\rho_1/2)^2\rho_1/2\\&+h_2(\rho_2/2)^3-h_2(\rho_1/2)^3-h_1(\rho_2/2)^3+h_1(\rho_1/2)^3\\
%-h_2(\rho_2/2)^2\rho_1/2 &\leq h_1(\rho_1/2)^2\rho_2/2-h_2(\rho_1/2)^3-h_1(\rho_2/2)^3\\
%-h_2(\rho_2/\rho_1)^2 &\leq h_1\rho_2/\rho_1-h_2-h_1(\rho_2/\rho_1)^3\\
%h_1(\rho_2/\rho_1)^3-h_1\rho_2/\rho_1 &\leq h_2(\rho_2/\rho_1)^2-h_2\\
%h_1(\rho_2/\rho_1)((\rho_2/\rho_1)^2-1) &\leq h_2((\rho_2/\rho_1)^2-1)\\
%%%%%  h_1\rho_2 \leq h_2\rho_1.\label{ratio2}
% \end{align}

%\begin{align*}
%\textrm{using that }\frac{g_2}{\rho_2/2}=\frac{g_1}{\rho_1/2}&=\frac{h_2-h_1}{\rho_2/2-\rho_1/2}\\
%\frac{\pi(\rho_2/2)^2g_2-\pi(\rho_1/2)^2g_1}{d}
%&= \frac{\pi(\rho_2/2)^2\frac{(h_2-h_1)\rho_2/2}{\rho_2/2-\rho_1/2}-\pi(\rho_1/2)^2\frac{(h_2-h_1)\rho_1/2}{\rho_2/2-\rho_1/2}}{d}\\
%&= \frac{\pi(\rho_2/2)^3(h_2-h_1)-\pi(\rho_1/2)^3(h_2-h_1)}{d(\rho_2/2-\rho_1/2)}\\
%&= \frac{\pi(\rho_2/2)^3-\pi(\rho_1/2)^3}{d(\rho_2/2-\rho_1/2)}(h_2-h_1)\\
%\end{align*}

Given that $\varepsilon \leq \alpha/e$, we know that $V_{d}(\rho_1) \geq 1/n$. Then, given that $\pi(\rho_1/2)^2(h_2-h_1)=1/n$, it holds that $h_2\leq 2h_1$. Replacing in Equation~\ref{ratio2} we obtain $\rho_2\leq 2\rho_1$.
%\textcolor{red}{------------------------ define notation, explain formally}
The base of the big triangle is $\rho_2/2+\rho_1/4$, and the height is $h$.
The base of the triangle to compute is $\rho_1/(2\sqrt{2})+\rho_1/4$, and the height is 
$h (\rho_1/(2\sqrt{2})+\rho_1/4)/(\rho_2/2+\rho_1/4)$.
%$$h \frac{\rho_1/(2\sqrt{2})+\rho_1/4}{\rho_2/2+\rho_1/4}$$
The base of the small triangle to substract is $\rho_1/2$, and the height is 
$h \rho_1/(2(\rho_2/2+\rho_1/4))$.
%$$h \frac{\rho_1}{2(\rho_2/2+\rho_1/4)}$$
Then, the trapezoid area is
\begin{align*}
\frac{3}{8}\rho_1 h \frac{\rho_1/(2\sqrt{2})+\rho_1/4}{\rho_2/2+\rho_1/4}-\frac{\rho_1}{4} h \frac{\rho_1}{2(\rho_2/2+\rho_1/4)}\\
= \frac{\rho_1^2}{2\rho_2+\rho_1} h \left(\frac{3}{2} \left(\frac{1}{2\sqrt{2}}+\frac{1}{4}\right) - \frac{1}{2}\right)\\
\geq \rho_1 h \frac{1}{4}\left(\frac{3}{2} \left(\frac{1}{\sqrt{2}}+\frac{1}{2}\right) - 1\right).
\end{align*}

Then, the volume of $B_a$ is at least
\begin{align*}
\rho_1^2 h \frac{1}{4\sqrt{2}}\left(\frac{3}{2} \left(\frac{1}{\sqrt{2}}+\frac{1}{2}\right) - 1\right)\\
=\frac{1}{\pi\sqrt{2}n}\left(\frac{3}{2} \left(\frac{1}{\sqrt{2}}+\frac{1}{2}\right) - 1\right)
\geq \frac{1}{6n}.
\end{align*}

Then, the probability that there is a point $a\in P$ located in $B_a$ is larger than $1-\left(1-\frac{1}{6n}\right)^{n}
\geq 1-e^{-1/6}$, by Inequality~\ref{taylor}.
%\begin{align*}
%> 1-\left(1-\frac{1}{6n}\right)^{n}
%&\geq 1-e^{-1/6},\textrm{ by Inequality~\ref{taylor}}.
%\end{align*}
%
And the probability that there is another point $b\in P$ located in $B_b$ is
\begin{align*}
1-\left(1-\frac{1}{6n}\right)^{n-1}
&\geq 1-e^{-(n-1)/(6n)}
\geq 1-e^{-1/12}, 
\end{align*}
where the last inequality holds for any $n>1$,
and the first inequality follows from Inequality~\ref{taylor}.
It remains to be shown that $\Gamma$ is void of points.
The probability that $\Gamma$ is empty is lower bounded by upper bounding the volume, i.e. taking $V_{d}(\rho) \leq V_{d}(\rho_2)\leq V_{d}(2\rho_1)+(V_{d}(\rho_2)-V_{d}(\rho_1))$. We know that $V_{d}(\rho_2)-V_{d}(\rho_1)\leq \pi(\rho_2/2)^2 h\leq \pi(\rho_1)^2 h=4/n$. Then, for $V_{d}(\rho_2)-V_{d}(\rho_1)$, we have
\begin{align*}
\left(1-\frac{4}{n}\right)^{n-2}
&\geq e^{-4(n-2)/(n-4)}
\geq 1/e^{12},\textrm{ for any $n>4$}
\end{align*}
(the first inequality follows from Inequality~\ref{taylor}),
and the probability that $\Gamma_1$ is empty is 
\begin{align*}
\left(1-V_{d}(\rho_1)\right)^{n-2} &\geq \exp\left(-\frac{V_{d}(\rho_1)(n-2)}{1-V_{d}(\rho_1)}\right),\textrm{ by Inequality~\ref{taylor}}.
\end{align*}

Replacing, we get
\begin{align*}
Pr\left((a,b)\in D(P)\right) \geq\\
 \left(1-\frac{1}{e^{1/6}}\right) \left(1-\frac{1}{e^{1/12}}\right)\frac{1}{e^{12}}  \exp\left(-\frac{V_{d}(\rho_1)(n-2)}{1-V_{d}(\rho_1)}\right)
 = \varepsilon.
\end{align*}

Which proves the first part of the claim.
For the second part, we upper bound the volume of the spherical cap of base diameter $\rho_1$ with the volume of the cylinder circumscribing it.
\begin{align*}
V_{3}(\rho_1) &\leq \frac{\pi\rho_1^2}{4}\left(\sqrt[3]{\frac{3}{4\pi}}-\sqrt{\left(\frac{3}{4\pi}\right)^{2/3}-\frac{\rho_1^2}{4}}\right).
\end{align*}
Hence,
\begin{align*}
\sqrt{\left(\frac{\pi}{4}\sqrt[3]{\frac{3}{4\pi}}\right)^2\rho_1^4-\frac{\pi^2}{64}\rho_1^6}
&\leq 
\frac{\pi\rho_1^2}{4}\sqrt[3]{\frac{3}{4\pi}}-V_{3}(\rho_1).
\end{align*}
Given that $\pi\rho_1^2/4\sqrt[3]{3/(4\pi)}\geq V_{3}(\rho_1)$, we can square both sides getting
\begin{align*}
\frac{\pi^2}{64}\rho_1^6
&\geq \left(2\frac{\pi\rho_1^2}{4}\sqrt[3]{\frac{3}{4\pi}}-V_{3}(\rho_1)\right)V_{3}(\rho_1)\\
&\geq \frac{\pi\rho_1^2}{4}\sqrt[3]{\frac{3}{4\pi}}V_{3}(\rho_1), \textrm{ because $V_{3}(\rho_1)\leq \pi\rho_1^2/4\sqrt[3]{3/(4\pi)}$}.
\end{align*}
Then we get $\rho_1/2 \geq \sqrt[4]{\sqrt[3]{48/\pi^4}V_{3}(\rho_1)}$ and replacing $V_{3}(\rho_1)$ the claim follows.

\end{proof}

\section{Future Directions, Open Problems}

It would be interesting to extend this study to other norms, such as
$L_1$ or $L_\infty$.  Also, Theorems~\ref{thm:2dballlb}
and~\ref{thm:3dballlb} were proved by showing that the existence of a configuration
that yields a Delaunay edge of some length is not unlikely. Different
configurations were used for each, but a configuration that works for
both cases exists (although yielding worse constants). We conjecture
that (modulo some constant) the same bound can be obtained in general
for any $d>1$.  Both questions are left for future work.

%---------------------------- Bibliography -------------------------------

%\small 
\bibliographystyle{abbrv}
\bibliography{./refs}

\end{document}